\documentclass[10pt, draftclsnofoot, onecolumn]{IEEEtran}
\usepackage{amsthm}
\usepackage{amsmath}
\usepackage{amsfonts}
\usepackage{graphicx}
\usepackage{latexsym}
\usepackage{amssymb}
\usepackage{stmaryrd}
\usepackage{comment}
\usepackage{amscd}
\usepackage{hyperref}
\newcommand\myshade{70} 
\hypersetup{ 
	linkcolor  = red!\myshade!black,
	citecolor  = blue!\myshade!black,
	urlcolor   = blue!\myshade!black,
	colorlinks = true,
}
\usepackage{cite}
\usepackage{cleveref}
\usepackage{colortbl}
\usepackage[belowskip=-15pt,aboveskip=0pt,small]{caption}
\usepackage[dvipsnames]{xcolor}
\usepackage[linesnumbered,ruled]{algorithm2e}
\usepackage[noend]{algpseudocode}
\allowdisplaybreaks
\usepackage{graphicx}
\usepackage{subfigure}
\usepackage{wrapfig}
\usepackage{float}
\usepackage{bm}
\usepackage{bbm}
\usepackage{enumerate}

\usepackage{thm-restate}

\usepackage{mathtools}
\newcommand{\bea}{\begin{eqnarray}}
\newcommand{\eea}{\end{eqnarray}}
\newcommand{\bean}{\begin{eqnarray*}}
\newcommand{\eean}{\end{eqnarray*}}
\newcommand{\ceil}[1]{\left\lceil #1 \right\rceil}
\newcommand{\floor}[1]{\left\lfloor #1 \right\rfloor}

\newcommand{\sbinom}[2]{\left( \begin{array}{c} #1 \\ #2 \end{array} \right) }

\newcommand{\cA}{{\cal A}}
\newcommand{\cB}{{\cal B}}

\newcommand{\cD}{{\cal D}}

\newcommand{\cL}{{\cal L}}
\newcommand{\cM}{{\cal M}}
\newcommand{\cN}{{\cal N}}

\newcommand{\cS}{{\cal S}}

\newcommand{\cU}{{\cal U}}

\newcommand{\sG}{\script{G}}

\newcommand{\sP}{\script{P}}

\newcommand{\bfc}{{\boldsymbol c}}

\newcommand{\bfm}{{\boldsymbol m}}

\newcommand{\bfx}{{\boldsymbol x}}

\newcommand{\bvert}{\big\vert}

\DeclareMathOperator*{\E}{\mathbf{E}}

\DeclareMathAlphabet{\mathbfsl}{OT1}{cmr}{bx}{it}
\newcommand{\uuu}{\kern-1pt\mathbfsl{u}\kern-0.5pt}
\newcommand{\vvv}{\kern-1pt\mathbfsl{v}\kern-0.5pt}

\newcommand{\myboxplus}{\kern1pt\mbox{\small$\boxplus$}}

\makeatletter \DeclareRobustCommand{\sbinom}{\genfrac[]\z@{}}
\makeatother
\newcommand{\G}[2]{\sbinom{{#1}\kern-1pt}{{#2}\kern-1pt}}
\newcommand{\Gq}[2]{\sbinom{{#1}\kern-0.25pt}{{#2}\kern-0.25pt}}
\newcommand{\Fq}{\smash{{\mathbb F}_{\!q}}}
\newcommand{\Fqn}{\smash{{\mathbb F}_{\!q}^n}}
\newcommand{\Fqk}{\smash{{\mathbb F}_{\!q}^k}}

\newcommand{\Ps}{\smash{{\sP\kern-2.0pt}_q\kern-0.5pt(n)}}
\newcommand{\sPs}{\smash{{\sP\kern-1.5pt}_q(n)}}
\newcommand{\Ptwo}{\smash{{\sP\kern-2.0pt}_2\kern-0.5pt(n)}}
\newcommand{\Ptwom}{\smash{{\sP\kern-2.0pt}_2\kern-0.5pt(m)}}
\newcommand{\Ptwonm}{\smash{{\sP\kern-2.0pt}_2\kern-0.5pt(n+m)}}
\newcommand{\Ptwoa}{\smash{{\sP\kern-2.0pt}_2\kern-0.5pt(1)}}
\newcommand{\Ptwob}{\smash{{\sP\kern-2.0pt}_2\kern-0.5pt(2)}}
\newcommand{\Ptwoc}{\smash{{\sP\kern-2.0pt}_2\kern-0.5pt(3)}}
\newcommand{\Ptwod}{\smash{{\sP\kern-2.0pt}_2\kern-0.5pt(4)}}
\newcommand{\Ptwoe}{\smash{{\sP\kern-2.0pt}_2\kern-0.5pt(5)}}
\newcommand{\Ptwof}{\smash{{\sP\kern-2.0pt}_2\kern-0.5pt(6)}}
\newcommand{\Ptwokm}{\smash{{\sP\kern-2.0pt}_2\kern-0.5pt(2k-1)}}
\newcommand{\Pone}{\smash{{\sP\kern-2.5pt}_2\kern-0.5pt(n{-}1)}}

\newcommand{\Gr}{\smash{{\sG\kern-1.5pt}_q\kern-0.5pt(n,k)}}
\newcommand{\Gi}{\smash{{\sG\kern-1.5pt}_q\kern-0.5pt(n,i)}}
\newcommand{\Gj}{\smash{{\sG\kern-1.5pt}_q\kern-0.5pt(n,j)}}
\newcommand{\Grmk}{\smash{{\sG\kern-1.5pt}_q\kern-0.5pt(n,n-k)}}
\newcommand{\Grdk}{\smash{{\sG\kern-1.5pt}_q\kern-0.5pt(2k,k)}}
\newcommand{\Grekappa}{\smash{{\sG\kern-1.5pt}_q\kern-0.5pt(n,e+1-\kappa)}}
\newcommand{\Grtwoekappa}{\smash{{\sG\kern-1.5pt}_q\kern-0.5pt(n,2e+1-\kappa)}}
\newcommand{\Gremkappa}{\smash{{\sG\kern-1.5pt}_q\kern-0.5pt(n,e-\kappa)}}
\newcommand{\Gn}{\smash{{\sG\kern-1.5pt}_2\kern-0.5pt(n,n{-}1)}}
\newcommand{\Gnq}{\smash{{\sG\kern-1.5pt}_q\kern-0.5pt(n,n{-}1)}}
\newcommand{\Gone}{\smash{{\sG\kern-1.5pt}_2\kern-0.5pt(n,1)}}
\newcommand{\Gqone}{\smash{{\sG\kern-1.5pt}_q\kern-0.5pt(n,1)}}
\newcommand{\GTwo}{\smash{{\sG\kern-1.5pt}_2\kern-0.5pt(n,k)}}
\newcommand{\GTwonk}[2]{{\smash{{\sG\kern-1.5pt}_2\kern-0.5pt({#1},{#2})}}}
\newcommand{\Gnk}{\smash{{\sG\kern-1.5pt}_2\kern-0.5pt(n,n{-}k)}}
\newcommand{\Greone}{\smash{{\sG\kern-1.5pt}_q\kern-0.5pt(n,e{+}1)}}
\newcommand{\Gretwo}{\smash{{\sG\kern-1.5pt}_q\kern-0.5pt(n,e{+}2)}}

\newcommand{\be}[1]{\begin{equation}\label{#1}}
\newcommand{\ee}{\end{equation}}

\newtheorem{lemma}{Lemma}
\newtheorem{remark}{Remark}
\newtheorem{corollary}{Corollary}

\newtheorem{definition}{Definition}
\newtheorem{proposition}{Proposition}

\newtheorem{conjecture}{Conjecture}

\newcommand{\Hent}[1]{\mathbf{H}\left(#1\right)}
\newcommand{\Iinfo}[2]{\mathbf{I}\left(#1; #2\right)}
\newcommand{\condH}[2]{\mathbf{H}\left(#1 \mid #2\right)}
\newcommand{\condI}[3]{\mathbf{I}\left(#1; #2 \mid #3\right)}

  \newcommand*{\Initial}{I\mskip-2mu}
\newcommand*{\Final}{F\mskip-2mu}
\newcommand*{\In}{n^{\Initial}}
\newcommand*{\Fn}{n^{\Final}}
\newcommand*{\Ik}{k^{\Initial}}
\newcommand*{\Fk}{k^{\Final}}

\newcommand*{\IPart}{{\mathcal{P}}^\Initial}
\newcommand*{\FPart}{{\mathcal{P}}^\Final}

\newcommand*{\Ir}{{r^{\Initial}}}
\newcommand*{\Fr}{{r^{\Final}}}

\newcommand*{\Il}{{\lambda^{\Initial}}}
\newcommand*{\Fl}{{\lambda^{\Final}}}

\newcommand{\bw}{\gamma_{\mathrm{R}}}

\IEEEoverridecommandlockouts

\begin{document}

\author{
    \IEEEauthorblockN{\textbf{Shubhransh~Singhvi}\IEEEauthorrefmark{1}, \textbf{Saransh~Chopra}\IEEEauthorrefmark{1}, \textbf{K.V.~Rashmi}\IEEEauthorrefmark{1}\IEEEauthorrefmark{2}\thanks{This work was supported in part by the NSF CAREER Award under Grant 19434090 and in part by a Sloan Faculty Fellowship.}\\}
    \IEEEauthorblockA{\IEEEauthorrefmark{1}Carnegie Mellon University}
    \IEEEauthorblockA{\IEEEauthorrefmark{2}Google
    \\Email: shubhranshsinghvi2001@gmail.com, \{saranshc, rvinayak\}@cs.cmu.edu}
}

\title{Tight Lower Bounds on the Bandwidth Cost of MDS Convertible Codes in the Split Regime}
\date{\today}
 \maketitle
\begin{abstract}
Recent advances in erasure coding for distributed storage systems have demonstrated that adapting redundancy to varying disk failure rates can lead to substantial storage savings. Such adaptation requires \emph{code conversion}, wherein data encoded under an initial $[\Ik + \Ir, \Ik]$ code is transformed into data encoded under a final $[\Fk + \Fr, \Fk]$ code—an operation that can be resource-intensive. \emph{Convertible codes} are a class of codes designed to facilitate this transformation efficiently while preserving desirable properties such as the MDS property. In this work, we investigate the fundamental limits on the \emph{bandwidth cost} of conversion (total amount of data transferred between the storage nodes during conversion) for systematic MDS convertible codes. Specifically, we study the subclass of conversions known as the \emph{split regime} (a single initial codeword is converted into multiple final codewords).

In this setting, prior to this work, the best known lower bounds on the bandwidth cost of conversion for all parameters were derived by Maturana and Rashmi~\cite{maturana2022bandwidth} under certain uniformity assumptions on the number of symbols downloaded from each node. Further, these bounds were shown to be tight for the parameter regime where $\Fr \geq \Fk$ or $\Ir \leq \Fr$. In this work, we derive lower bounds on the bandwidth cost of systematic MDS convertible codes for all parameters in the split regime without the uniformity assumption. Moreover, our bounds are tight for the broader parameter regime where $\Fr \geq \Fk$ or $\Ir \leq \Fk$. Subsequently, our bounds also partially resolve the conjecture proposed in \cite{maturana2022bandwidth}. We employ a novel information-theoretic framework, which assumes only that the initial and final codes are systematic and does not rely on any linearity assumptions or the aforementioned uniformity assumptions.
\end{abstract}

\section{Introduction}
Distributed data storage systems (DSS) are at the core of modern computing infrastructure, providing data reliability, high availability, and efficient access to meet growing scalability demands. A critical enabler of these systems is the use of erasure codes, which store large amounts of data reliably while minimizing storage overhead~\cite{patterson1988case, ghemawat2003google, huang2012erasure}.
An $[n, k]$ \emph{erasure code} encodes $k$ data symbols into a \emph{codeword} consisting of $n$ symbols, which are distributed across multiple storage devices (nodes). Large-scale DSSs typically consist of several independent codewords distributed across different subsets of storage nodes.

Among erasure codes, maximum distance separable (MDS) codes are of particular interest. For a given $(n, k)$ configuration, MDS codes achieve the optimal trade-off between storage overhead and failure tolerance by tolerating up to $n-k$ node failures (erasures), the maximum possible due to the Singleton bound~\cite{singleton1964maximum}. Alternatively, an $(n, k)$-MDS coded DSS possesses the property that the entire data can be recovered from any $k$ out of the $n$ nodes. This property is referred to as the \emph{MDS property}.

Due to the inherent trade-off between storage overhead and failure tolerance, selecting the best possible $(n, k)$ configuration is crucial and must be carefully aligned with quality of service (QoS) requirements.
 Recent findings by Kadekodi et al.~\cite{kadekodi2019cluster} reveal that the failure rates of storage devices in a DSS vary over time and exhibit a “bathtub” shape with three distinct phases of life: failure rates are high during the “infancy” phase (1–3 months), low and stable during the “useful life” phase (3–5 years), and high again during the “wearout” phase (a few months before decommissioning). Their research highlights the importance of adapting the erasure code parameters of already encoded data to the changing reliability characteristics of storage devices. Additionally, adapting these parameters for data with varying access patterns, such as \emph{hot data}, which is frequently accessed, and \emph{cold data}, which is accessed infrequently, can significantly optimize performance and resource usage~\cite{kim2024morph}.

Whenever the parameters $n$ and $k$ are updated, all the data that is already encoded must be modified to conform to the newly chosen parameters. However, the \emph{default approach} of re-encoding all data is prohibitively expensive~\cite{maturana2022convertible}. To address these challenges, Francisco and Rashmi~\cite{maturana2022convertible} introduced the \emph{code conversion problem}, which provides a theoretical framework for efficiently updating the code parameters of already encoded data. Specifically, code conversion refers to the process of transforming multiple codewords of an initial code $\mathcal{C}^I$ with parameters $[\In, \Ik]$ into multiple codewords of a final code $\mathcal{C}^F$ with parameters $[\Fn, \Fk]$. Let $\Ir\;\triangleq\;\In - \Ik$ and $\Fr\;\triangleq\;\Fn - \Fk$ represent the number of redundant symbols (aka parity symbols) of the initial and final codewords, respectively.

The primary objective in the study of code conversion is to design the initial and final codes, along with a conversion procedure, that enables the transformation of encoded data more efficiently than the default re-encoding approach, for given parameters $(\In, \Ik; \Fn, \Fk)$. This design is subject to decodability constraints on the codes $\mathcal{C}^I$ and $\mathcal{C}^F$, such as both satisfying the MDS property, since these codes represent encoded data at different points in time within the storage system. 
A pair of codes that enables efficient conversion from an $[\In, \Ik]$ code to an $[\Fn, \Fk]$ code is referred to as an $(\In, \Ik; \Fn, \Fk)$ \emph{convertible code}, and the initial $[\In, \Ik]$ code is said to be \emph{$(\Fn, \Fk)$-convertible}. Furthermore, an $(\In, \Ik; \Fn, \Fk)$ convertible code is said to be MDS if both the initial and final codes are MDS. It is said to be systematic if both the initial and final codes are systematic. Works on convertible codes have primarily considered the following two cost metrics: 
\begin{enumerate}
    \item The \emph{access cost} of a conversion procedure is the sum of the \emph{read access cost}, which refers to the total number of initial nodes read, and the \emph{write access cost}, which refers to the total number of final nodes written.

    \item The \emph{bandwidth cost} of the conversion procedure is the sum of the \emph{read bandwidth cost}, which is the total amount of data read from the initial nodes, and the \emph{write bandwidth cost}, which is the total amount of data written into the final nodes.
\end{enumerate}

A convertible code is termed \emph{access-optimal} or \emph{bandwidth-optimal} if its conversion procedure achieves the minimum access cost or bandwidth cost, respectively, among all convertible codes with parameters $(\In, \Ik; \Fn, \Fk)$. While tight lower bounds on the access cost of linear MDS convertible codes are well-established in the literature \cite{maturana2020access, maturana2022convertible}, lower bounds on the bandwidth cost remain an open problem and are the focus of this paper.

{\textbf{Scope of this work:} In this work, we investigate the fundamental limits of MDS code conversion in the \emph{split regime}, where a single initial codeword is split into multiple final codewords (i.e., $\Ik = \Fl \Fk$ for some integer $\Fl \geq 2$). In~\cite{maturana2022bandwidth}, Maturana and Rashmi designed linear MDS convertible codes for this regime, where both the initial and final codes, as well as the conversion procedure, are linear. Furthermore, assuming that same number of symbols are downloaded from each information node and same number of symbols are downloaded from each parity node--the authors derived a matching lower bound on the conversion bandwidth in the \emph{increasing redundancy regime} ($\Ir \leq \Fr$), and further \emph{conjectured} that their constructions are bandwidth-optimal in the \emph{decreasing redundancy regime} ($\Ir > \Fr$) as well. Under the uniformity assumption, they also showed that no improvement is possible over the default approach when $\Fr \geq \Fk$.

In this work, we remove the uniformity assumption on the number of symbols downloaded per node and establish lower bounds on the bandwidth cost of systematic MDS convertible codes in the split regime. Our approach is information-theoretic: we model the message symbols, as well as the data stored and transmitted by the nodes during conversion, as random variables and derive entropy-based bounds.

We begin by deriving a lower bound on the bandwidth cost of systematic MDS convertible codes for all parameters in the split regime (Lemma~\ref{lem:bw_lb_trivial}). Furthermore, we establish that this bound is tight when $\Fr \geq \Fk$, and that no improvement is possible over the default approach. Next, we derive a second lower bound on the bandwidth cost of systematic MDS convertible codes in the split regime where $\Fr < \Fk$ (Lemma~\ref{lem:bw_lb_I}). We establish that this bound is tight when $\Ir \leq \Fr$. Finally, we derive a third lower bound on the bandwidth cost of systematic MDS convertible codes in the split regime where $\Fr < \Ir \leq \Ik$ and $\Fr < \Fk$ (Lemma~\ref{lem:bw_lb_II}). For the regime $\Fr < \Ir \le \Fk$, we establish that this bound is tight and resolve the conjecture posed by Maturana and Rashmi in~\cite{maturana2022bandwidth} for systematic codes (see Conjecture~\ref{conj:split_regime}). We compare and summarize our bounds in Theorem~\ref{thm:MDS_split_lb}.

\textbf{Outline of the paper:}
Section~\ref{sec:background} reviews the necessary background on convertible codes. In Section~\ref{sec:MDS_Storage}, we establish fundamental information-theoretic inequalities for MDS coded storage, which serve as essential tools for our analysis. Section~\ref{sec:bandwidth_cost} derives lower bounds on the bandwidth cost of systematic MDS convertible codes in the split regime and identifies the parameter regimes in which these bounds are tight. Finally, Section~\ref{sec:conclusion} summarizes the main results and discusses potential directions for future research.
}

\section{Background, Notations and Related Work}\label{sec:background}

Convertible codes are a class of erasure codes which are designed to be efficiently tunable while retaining desirable properties such as being systematic and MDS~\cite{chopra2024low, maturana2020access, maturana2022bandwidth, maturana2023bandwidth, ge2024mds, ge2025locally, ramkumar2025mds} or having a local repair property~\cite{maturana2023LRCC, Kong2023LocallyRC, shi2025bounds}. Before formally defining convertible codes, we introduce some notations and basic definitions.

For integers $a, i \in \mathbb{Z}_{\geq 1}$, let $[a]^i$ denote the set of integers $[a]^i\;=\;\{a(i-1)+1, a(i-1)+2, \ldots, ai\}$ and let $[a]\,\triangleq\,[a]^{1}$. For an $n$-length vector $\bfx$ and a subset $\cA \subseteq [n]$, the notation $\bfx_{\cA}$ represents the projection of $\bfx$ onto the indices specified by the set $\cA$. For a prime power $q$, let $\mathbb{F}_q$ denote the finite field of size $q$. For any sets $\cA \subseteq \cB$, and any collection of random variables $\{X_i\}_{i \in \cB}$, let $X_\cA$ denote $\{X_i\}_{i \in \cA}$.

\begin{definition}A \emph{code} of block length $n$ over a finite field $\Fq$ is a subset 
$\mathcal{C} \subseteq \Fqn$.  
Equivalently, a code may be described as an injective mapping
\[
    \mathcal{C}: \mathcal{M} \to \Fqn,
\]
where $\mathcal{M}$ is the message set.
\end{definition}

For the ease of exposition, we assume that the set of messages $\cM = \Fqk$, for some positive integer $k$.

\begin{definition}
Let $\mathcal{C}:\Fqk \to \Fqn$ be a code.
We say that $\mathcal{C}$ is an $[n,k]_q$ \emph{systematic code} if there exists a fixed set of $k$ coordinates 
$S \subseteq [n]$ such that, for every $\bfm \in \Fqk$, the projection of 
$\mathcal{C}(\bfm)$ onto $S$ equals $\bfm$.
\end{definition}

\begin{definition}  
An $[n,k,\alpha]_q$ \emph{vector code} is an injective map
\[
    \mathcal{C} : \mathbb{F}_q^{\alpha k} \to \mathbb{F}_q^{\alpha n}.
\]
For a message vector $\bfm \in \mathbb{F}_q^{\alpha k}$, let the codeword $\bfc = \mathcal{C}(\bfm)$. Then, for any coordinate $i \in [n]$, 
the $i$-th \emph{symbol} of $\bfc$ is the block $\bfc_i = (c_{\alpha(i-1)+1},\ldots,c_{\alpha i}) \in \mathbb{F}_q^{\alpha}$,
whose entries are called \emph{subsymbols}.
\end{definition}

\begin{definition}[Convertible Code {\cite{maturana2020access}}]  
A \emph{convertible code} with parameters $(\In,\Ik;\Fn,\Fk)$ consists of:
\begin{enumerate}
    \item An \emph{initial code} $\mathcal{C}^I : \mathbb{F}_q^{\alpha \Ik} \to \mathbb{F}_q^{\alpha \In}$ 
          and a \emph{final code} $\mathcal{C}^F : \mathbb{F}_q^{\alpha \Fk} \to \mathbb{F}_q^{\alpha \Fn}$.
    \item A pair of partitions $\IPart$ and $\FPart$ of the common message index set 
          $[M]$ with $M = \mathrm{lcm}(\Ik,\Fk)$, such that each block in $\IPart$ has size $\Ik$ 
          and each block in $\FPart$ has size $\Fk$.
    \item A  \emph{conversion procedure} that, for all 
          $\bfm \in \mathbb{\Fq}^M$, takes as input
          $\{\mathcal{C}^I(\bfm_P) : P \in \IPart\}$ and produces as output
          $\{\mathcal{C}^F(\bfm_P) : P \in \FPart\}$.
\end{enumerate}
\end{definition}

During the conversion procedure, storage nodes can be categorized into three distinct types:  
\begin{enumerate}  
    \item \textbf{Unchanged nodes}: These are storage nodes that remain part of the final codewords without modification. No conversion bandwidth is required for these nodes since they are retained as is.  
    \item \textbf{Retired nodes}: These are the initial storage nodes that do not appear in the final codewords and are therefore removed during the conversion procedure.  
    \item \textbf{New nodes}: These are the additional storage nodes introduced in the final codewords, which were not part of the initial codewords.  
\end{enumerate} 

The \emph{conversion coordinator}, which orchestrates the conversion process, downloads data from both unchanged and retired nodes to reconstruct the new nodes. Convertible codes that maximize the number of unchanged nodes are referred to as \emph{stable codes}.

\begin{remark}
\label{rem:sys_unchanged}
For systematic convertible codes, both the initial and final codes are systematic and encode the same data. This implies the following:

\begin{enumerate}
    \item The information nodes of the initial codewords remain unchanged during the conversion procedure. These nodes directly serve as the information nodes in the final codewords, thereby eliminating any need to rewrite the information nodes.
    \item Consequently, any new node generated during the conversion procedure must be a final parity node.
    \item Similarly, any node retired during the conversion procedure must be an initial parity node.
\end{enumerate}
\end{remark}

{
\subsection{Conversion Cost of MDS Convertible Codes}

Numerous previous works have studied the read access and read bandwidth costs of MDS convertible codes \cite{maturana2020access,maturana2023bandwidth,maturana2022bandwidth}. Let $d_{\mathrm{R}}$ denote the read access cost of any linear MDS $(\In, \Ik; \Fn, \Fk)$ convertible code in the general regime, with no restrictions on the code parameters. When $\Ik \neq \Fk$, Maturana and Rashmi~\cite{maturana2020access} established the following lower bound:
\begin{align*}
    d_{\mathrm{R}} \ge \begin{cases}
        \Il \Ik &  \Ir < \Fr~\text{or}~\Fr \ge \min\{\Ik,\Fk\}\;,\\
        \Il \Fr + (\Il \pmod{\Fl})(\Ik - \max\{\Fk \pmod{\Ik}, \Fr\}) & \text{otherwise}\;.
    \end{cases}
\end{align*} 
Furthermore, they demonstrated the tightness of this bound by constructing linear MDS convertible codes that achieve the minimum possible read access cost.

In~\cite{maturana2023bandwidth}, Maturana and Rashmi presented a lower bound on the bandwidth cost of MDS convertible codes in the ``merge regime'', where $\Fk = \Il \Ik$ for $\Il \ge 2$. They established that for any MDS $(\In, \Ik; \Fn, \Fk = \Il \Ik)$ convertible code, the read bandwidth cost, denoted by $\bw$, satisfies
\[
\bw \ge 
\begin{cases}
\Il \alpha \min\{\Ik, \Fr\} & \text{if } \Ir \ge \Fr \text{ or } \Ik \le \Fr\;, \\[8pt]
\Il \alpha \left(\Ir + \Ik \!\left(1 - \dfrac{\Ir}{\Fr}\right)\right) & \text{otherwise}\;.
\end{cases}
\]
The authors demonstrated the tightness of this bound by providing constructions of MDS convertible codes for all parameters in the merge regime achieving this bandwidth cost.

\begin{definition}[Uniform Cost Assumption]\label{def:uni_cost}
Under this assumption, the conversion coordinator downloads same number of symbols (say, $\beta$) from each information node and same number of symbols (say, $\sigma$) from each parity node during the conversion procedure. 
\end{definition}

Under the Uniform Cost Assumption, Maturana and Rashmi~\cite{maturana2022bandwidth} presented a lower bound on the bandwidth cost of MDS convertible codes in the split regime, where $\Ik = \Fl \Fk$ for $\Fk \ge 2$. They showed that for any stable MDS $(\In, \Ik = \Fl \Fk; \Fn, \Fk)$ convertible code, the read bandwidth cost satisfies
\[
\bw \ge 
\begin{cases}
    \Fl\Fk\alpha-\Ir\alpha\max\{\frac{\Fk}{\Fr}-1,0\} & \text{if } \Ir \leq \Fl\Fr, \\[8pt]
    \Fl\min\{\Fr,\Fk\}\alpha & \text{otherwise}\;.
\end{cases}
\]
The authors proved this bound to be tight (under the Uniform Cost Assumption) when either $\Fr \ge \Fk$, in which case it is achievable by the default approach, or $\Ir\le\Fr<\Fk$, with an explicit construction. For the regime in which $\Fr < \Fk$ and $\Fr < \Ir$, the authors proved that this bound is not always achievable and therefore not tight; moreover, they presented constructions of systematic MDS $(\In, \Ik = \Fl \Fk; \Fn, \Fk)$ convertible codes in this regime achieving a read bandwidth cost of 
\[
\bw = \Fl\Fr\alpha\frac{(\Fl-1)\Fk + \Ir}{(\Fl-1)\Fr + \Ir}\;.
\]
and posed the following conjecture which would imply these constructions are bandwidth-optimal under the Uniform Cost Assumption.

\begin{conjecture}\label{conj:split_regime}\cite[Theorem 7]{maturana2022bandwidth}. 
Under the Uniform Cost Assumption, all stable linear MDS $(\In, \Ik\;=\;\Fl \Fk; \Fn, \Fk)$ convertible codes with $\Fr < \Ir$ and $\Fr < \Fk$  satisfy: 
\begin{align*}
    \bw\;\geq\;\Fl \Fr \alpha \frac{(\Fl - 1)\Fk + \Ir}{(\Fl - 1)\Fr + \Ir}\;,
\end{align*}
    where $\bw$ represents the read bandwidth cost. 
\end{conjecture}

}

\subsection{Other Related Work}

In \cite{Justin2025}, Zhang and Rashmi  investigated the design of access-optimal convertible codes that provide information-theoretic security in the presence of passive eavesdroppers. 

A related problem is the \emph{scaling problem}~\cite{
zhang2010alv,
zheng2011fastscale,
wu2012gsr,
zhang2014rethinking,
huang2015scaleRS,
wu2016ioefficient,
zhang2018optimal,
hu2018generalized,
zhang2020efficient,
rai2015adaptive,
rai2015adaptive2,
wu2020optimal,
hu2021combinedlocality,
wu2022optimaltradeoff,
wu2022placement
}, which involves converting each codeword of an $[n, k, \alpha]$ code into a codeword of an $[n + s, k + s, \frac{k\alpha}{k + s}]$ code, for a given integer $s$. That is, the total amount of data stored in a codeword is preserved, but it is redistributed across a smaller or larger number of storage nodes. Another related problem is the \emph{data rebalancing problem}~\cite{CDR2020,CDRDecentralized2020, CDRCyclic2022}, which studies the problem of efficient redistribution of replicated data when nodes are added to or removed from a storage system.

In this paper, we present an information-theoretic approach where the data stored and transferred are represented as random variables, and use entropy and mutual information based arguments to derive the lower bounds on conversion bandwidth. A similar information-theoretic approach has been used to prove the non-achievability of the storage-repair-bandwidth tradeoff in regenerating codes~\cite{shah2011distributed}.

\section{Information-theoretic inequalities for MDS coded storage}
\label{sec:MDS_Storage}

In this section, we first restate the MDS property in information-theoretic terms and then establish a key structural result for $(n,k)$-MDS coded distributed storage systems, which will serve as the foundation for our subsequent analysis. Consider an $(n, k)$-MDS coded distributed storage system (DSS), where the MDS code is systematic. Let the DSS comprise of several independent codewords. In what follows, we consider one such codeword. Such a codeword consists of
\begin{enumerate}
  \item $k$ \emph{information nodes}, storing the message symbols directly,
  \item $r\,=\,n - k$ \emph{parity nodes}, each of which is a deterministic function of the $k$ information nodes.
\end{enumerate}
Denote by $X_1,\dots,X_k$ the random variables corresponding to the data stored in the information nodes, and by $Y_1,\dots,Y_r$ those in the parity nodes. Since each parity node is a deterministic function of the information nodes, we have
\[
\condH{Y_{[r]}}{X_{[k]}}\;=\;0\;.
\]
The MDS property further implies that  for every pair of index‐sets $\cA\subseteq[r]$ and $\cB\subseteq[k]$ with $\lvert \cA\rvert + \lvert \cB\rvert\;=\;k$,
\[
\condH{X_{[k]}}{Y_{\cA},\,X_{\cB}}\;=\;0\;.
\]
We assume each node can store at most $\alpha$ symbols over $\Fq$. Hence
\[
\Hent{X_i}\;\leq\;\alpha,\quad i\in[k]\;,
\qquad
\Hent{Y_j}\;\leq\;\alpha,\quad j\in[r]\;.
\]
Moreover, we assume that the message symbols are independent and uniformly distributed, so in fact
\[
\Hent{X_i}\;=\;\alpha,\quad i=1,\dots,k\;.
\]

\begin{proposition}\label{prop:joint_entropy}
The total entropy of the data stored across the $k$ information nodes and the $r$ parity nodes is
\[
    \Hent{X_{[k]},Y_{[r]}}\;=\;k\alpha\;.
\]
\end{proposition}

\begin{proof}
Since each parity node’s data is a deterministic function of the information nodes $X_{[k]}$, we have
\begin{align*}
    \Hent{X_{[k]}, Y_{[r]}}\;=\;\Hent{X_{[k]}}\;.
\end{align*}
Moreover, $X_{[k]}$ are independent and each has entropy $\alpha$, so
\begin{align*}
    \Hent{X_{[k]}}\;=\;\sum_{i=1}^{k} \Hent{X_{i}}\;=\;k\alpha\;.
\end{align*}
Combining these yields
\begin{align*}
    \Hent{X_{[k]}, Y_{[r]}}\;=\;k\alpha\;,
\end{align*}
as required.
\end{proof}

\begin{proposition}\label{Prop:iud_parity}
Let $\cA \subseteq [r]$ be any subset of parity‐node indices with $\lvert \cA\rvert\,\leq\,k$.  Then the random variables $Y_{\cA}$ are independent and uniformly distributed.
\end{proposition}

\begin{proof}
Choose any subset $\cB \subseteq [k]$ such that $\lvert \cA\rvert + \lvert \cB\rvert\;=\;k$.  By the MDS property,
\[
\condH{X_{[k]}}{Y_{\cA},X_{\cB}}\;=\;0\;.
\]
On the other hand, the chain rule of entropy and Proposition~\eqref{prop:joint_entropy} give
\[
\Hent{X_{[k]}} 
\;=\;\Hent{Y_{\cA},X_{\cB}} 
+ \condH{X_{[k]}}{Y_{\cA},X_{\cB}}
\;=\;\Hent{Y_{\cA},X_{\cB}}\;,
\]
and since $X_{[k]}$ are independent uniform, $\Hent{X_{[k]}}\;=\;k\alpha$.  Hence
\[
\Hent{Y_{\cA},X_{\cB}}\;=\;k\alpha\;.
\]
By subadditivity,
\[
\Hent{Y_{\cA},X_{\cB}}
\;\leq\;\Hent{Y_{\cA}} + \Hent{X_{\cB}} \;=\;\Hent{Y_{\cA}} + \lvert\cB\rvert\,\alpha\;,
\]
so
\[
\Hent{Y_{\cA}}\;\geq\;\lvert\cA\rvert\,\alpha\;.
\]
But each parity node has storage capacity $\alpha$, so
\[
\Hent{Y_{\cA}}\;\leq\;\sum_{i\in\cA}\Hent{Y_i}\;\leq\;\lvert\cA\rvert\,\alpha\;.
\]
Thus $\Hent{Y_{\cA}}\;=\;\lvert\cA\rvert\,\alpha$, and equality in subadditivity implies that $Y_i$, $i\in\cA$, are independent.  Moreover, each $Y_i$ attains its maximum entropy $\alpha$, so they are uniformly distributed.
\end{proof}

For each $i \in [r]$, let $\mu_i$ be a deterministic function of the random variable $Y_i$.  Similarly, for each $j \in [k]$, let $\nu_j$ be a deterministic function of $X_j$.  Since these mappings are deterministic, we have
\begin{align*}
    \condH{\nu_j(X_j)}{X_j}\;=\;0,
    &\quad
    \Hent{\nu_j(X_j)} \;\triangleq\; \beta_j \;\leq\;\alpha\,,\\
    \condH{\mu_i(Y_i)}{Y_i}\;=\;0,
    &\quad
    \Hent{\mu_i(Y_i)} \;\triangleq\; \sigma_i \;\leq\;\alpha\,.    
\end{align*}

For any subsets $\cA \subseteq [r]$ and $\cB \subseteq [k]$, we abuse notation slightly and write
\[
    \mu_\cA(Y_\cA)\;=\;\bigl\{\mu_i(Y_i)\bigr\}_{i\in\cA}\;,
    \quad
    \nu_\cB(X_\cB)\;=\;\bigl\{\nu_j(X_j)\bigr\}_{j\in\cB}\;.
\]

The following lemma establishes a general upper bound on the mutual information between functions of random variables, which will be useful in our later analysis.

\begin{lemma}\label{lem:mutual_information_bound}
    Let $Z_1, Z_2, \dots,Z_n$ be a collection of random variables. Given any deterministic functions $\{f_i\}_{i \in [n]}$, disjoint subsets $A, B \subseteq [n]$, and subsets $\cD_1 \subseteq A,\cD_2 \subseteq B$ such that $Z_{(A \cup B)\setminus(\cD_1 \cup \cD_2)}$ are independent,
    $$\Iinfo{f_A(Z_A)}{f_B(Z_B)}\;\le\; \Hent{f_{\cD_1}(Z_{\cD_1})} + \Hent{f_{\cD_2}(Z_{\cD_2})}\;.$$
\end{lemma}

\begin{proof}
Observe that, as desired,
\begin{align*}
    \Iinfo{f_A(Z_A)}{f_B(Z_B)} 
    \;=\;&\Hent{f_A(Z_A)} + \Hent{f_B(Z_B)} - \Hent{f_{A \cup B}(Z_{A \cup B})}\\
    \;\leq\;&\Hent{f_{\cD_1}(Z_{\cD_1})} + \Hent{f_{\cD_2}(Z_{\cD_2})} 
    + \Hent{f_{A\setminus\cD_1}(Z_{A\setminus\cD_1})} \tag{Sub-additivity} \\
    &\qquad  + \Hent{f_{B\setminus\cD_2}(Z_{B\setminus\cD_2})} - \Hent{f_{A \cup B}(Z_{A \cup B})}\\
    \;=\;&\Hent{f_{\cD_1}(Z_{\cD_1})} + \Hent{f_{\cD_2}(Z_{\cD_2})} 
    \tag{Independence}\\
    &\qquad + \Hent{f_{(A \cup B)\setminus(\cD_1 \cup \cD_2)}(Z_{(A \cup B)\setminus(\cD_1 \cup \cD_2)})} - \Hent{f_{A \cup B}(Z_{A \cup B})}\\
    \;\leq\;&\Hent{f_{\cD_1}(Z_{\cD_1})} + \Hent{f_{\cD_2}(Z_{\cD_2})}\;.
\end{align*}
\end{proof}

Next, we prove the result that the minimum possible entropy of an $a$-set out of a set of $b$ random variables is at most the average entropy over all $a$-sets, given that these $a$-sets are all independent. We will use this in conjunction with the previous lemma to prove an upper bound on the mutual information between functions of data stored in information nodes and parity nodes.

\begin{lemma}\label{lem:min_avg_entropy_bound}
    Let $Z_1,\dots,Z_b$ be a collection of random variables and let $\{f_i\}_{i \in [b]}$ be deterministic. For any $a \le b$, where $\cA_a$ is the set of all subsets of $[b]$ of size $a$, if $Z_{A}$ are independent for all $A \in \cA_a$, $$\underset{A \in \cA_a}{\min} \;\Hent{f_A(Z_A)}\;\le\; \frac{a}{b}\,\sum_{i \in [b]}\Hent{f_{i}(Z_{i})}\;.$$
    If $Z_1,\dots,Z_b$ are also independent, then $$\underset{A \in \cA_a}{\min} \;\Hent{f_A(Z_A)}\;\le\; \frac{a}{b}\,\Hent{f_{[b]}(Z_{[b]})}\;.$$
\end{lemma}

\begin{proof}
For any distribution $\cD$ over the set of subsets $\cA_a$, it holds that: 
\begin{align*}
    \underset{A \in \cA_a}{\min}\; \Hent{f_A(Z_A)}&\;\le\;\E_{A\sim\cD(\cA_a)}[\Hent{f_A(Z_A)}]\;.
\end{align*}
In particular, let $\cD$ be the uniform distribution over $\cA_a$. Then, 
\begin{align}
        \underset{A \in \cA_a}{\min}\; \Hent{f_A(Z_A)}&\;\le\;\E_{A \sim\mathsf{Unif}(\cA_a)}[\Hent{f_A(Z_A)}]\nonumber\\
        &\;=\;\E_{A \sim\mathsf{Unif}(\cA_a)}[\sum_{i\in A}\Hent{f_i(Z_i)}]\tag{Independence of $Z_{A}$}\\
        &\;=\;\E_{A \sim\mathsf{Unif}(\cA_a)}[\sum_{i\in [b]}\mathbb{I}_{i \in A}\cdot\Hent{f_i(Z_i)}]\nonumber\\
        &\;=\;\sum_{i\in [b]}\left(\Hent{f_i(Z_i)}\cdot\E_{A \sim\mathsf{Unif}(\cA_a)}[\mathbb{I}_{i \in A}]\right)\,\tag{Linearity of expectation}\\
        &\;=\;\frac{a}{b}\sum_{i\in [b]}\,\Hent{f_i(Z_i)}\;.\nonumber
    \end{align}
\end{proof}

We combine Lemma~\ref{lem:mutual_information_bound} and Lemma~\ref{lem:min_avg_entropy_bound} to produce an upper bound on the mutual information between $\mu_{[r]}(Y_{[r]})$ and $\nu_{[k]}(X_{[k]})$, which will be used in the subsequent section. We do so by considering all $\beta_1$-sets of a given subset $\cS_1 \subseteq [r]$ and $\beta_2$-sets of a given subset $\cS_2 \subseteq [k]$, given that $\beta_1 + \beta_2 = r$, which are in turn $\beta_1$-sets of $[r]$ and $\beta_2$-sets of $[k]$, respectively.

\begin{corollary}\label{cor:mutual_information_avg_bound_2_variable}
    Given any $\cS_1 \subseteq [r]$, $\cS_2 \subseteq [k]$, $\beta_1 \leq |\cS_1|$, and $\beta_2 \leq |\cS_2|$ such that $\beta_1 + \beta_2 = r$, 
\begin{align*}
    \Iinfo{\mu_{[r]}(Y_{[r]})}{\nu_{[k]}(X_{[k]})}&\;\le\;\underset{\substack{\cD_1 \subseteq \cS_1\\|\cD_1| = \beta_1}}{\min}\Hent{\mu_{\cD_1}(Y_{\cD_1})} + \underset{\substack{\cD_2 \subseteq \cS_2\\|\cD_2| = \beta_2}}{\min}\Hent{\nu_{\cD_2}(X_{\cD_2})} \\
   &\;\le\;\frac{\beta_1}{|\cS_1|}\sum_{i \in \cS_1}\Hent{\mu_{i}(Y_{i})} + \frac{\beta_2}{|\cS_2|}\Hent{\nu_{\cS_2}(X_{\cS_2})}\;.
\end{align*}
\end{corollary}

We also present a version of Corollary~\ref{cor:mutual_information_avg_bound_2_variable} that restricts $\cS_1$ to be the empty set, which will come in handy in the proofs in the subsequent section.

\begin{corollary}\label{cor:mutual_information_avg_bound_1_variable}
    Given any subset $\cS \subseteq [k]$ such that $r \leq |\cS|$, it follows that
\begin{align*}    
    \Iinfo{\mu_{[r]}(Y_{[r]})}{\nu_{[k]}(X_{[k]})}&\;\le\; \underset{\substack{\cD \subseteq \cS\\|\cD| = r}}{\min}\;\Hent{\nu_\cD(X_{\cD})}\\
    &\;\le\; \frac{r}{|\cS|}\Hent{\nu_{\cS}(X_{\cS})}\;.
\end{align*}

\end{corollary}

\section{Bandwidth Cost of MDS Code Conversion in the Split Regime}
\label{sec:bandwidth_cost}

In this section, we derive lower bounds on the bandwidth cost of systematic $(\In, \Ik=\Fl\Fk; \Fn, \Fk)$-MDS convertible codes. Recall that in the split regime, a single initial codeword is split into multiple final codewords. Let $\Ir\;\triangleq\;\In -\Ik$ and $\Fr\;\triangleq\;\Fn -\Fk$ denote the initial and final number of parity nodes, respectively. 

As mentioned in Remark \ref{rem:sys_unchanged}, for systematic convertible codes, the initial information nodes remain unchanged during the conversion procedure and serve as the information nodes in the final codewords. We now outline the key components and assumptions as follows:

\noindent\paragraph*{Information Nodes}
\begin{enumerate}
    \item Let the set of $\Ik$ information nodes be indexed by $s_{[\Ik]}$. For $j\in[\Ik]$, let $X_{j}$ denote the random variable corresponding to the data stored in the information node indexed by $s_j$. We assume that the random variables $X_{[\Ik]}$ are independent and uniformly distributed.
\end{enumerate}

\noindent\paragraph*{Initial Parity Nodes}
\begin{enumerate}
    \item  Let the set of $\Ir$ initial parity nodes be indexed by $p^I_{[\Ir]}$. For $j \in[\Ir]$, let $Y_{j}^I$ denote the random variable corresponding to the data stored in the parity node indexed by $p^I_{j}$. These random variables are deterministically generated from the data stored in information nodes, and satisfy:
    \[
    \condH{Y^I_{[\Ir]}}{X_{[\Ik]}}\;=\;0\;.
    \]

    \item For any subsets $\cA \subseteq [\Ir]$ and $\cB \subseteq [\Ik]$ such that $\lvert\cA\rvert + \lvert\cB\rvert\;=\;\Ik$, the MDS property ensures: 
    \[
\condH{X_{[\Ik]}}{Y_{\cA}^I, X_{\cB}}\;=\;0\;.
    \]
\end{enumerate}

\noindent\paragraph*{Final Parity Nodes}
Let $t \in [\Fl]$ index the final codewords. 
\begin{enumerate}
    \item  Let the information nodes indexed by ${s_{[\Fk]^t}}$ be the information nodes of the $t$-th codeword. Recall that $[\Fk]^t$ denotes the set $\{(t-1)\Fk + 1, \ldots t\Fk\}$. 

     \item Let the $\Fr$ final parity nodes of the $t$-th codeword be indexed by $p^F_{[\Fr]^t}$. For $j \in [\Fr]^t$, let $Y_{j}^F$  denote the random variable corresponding to the data stored in parity node indexed by $p^F_j$. These parity nodes are fully determined by the corresponding $\Fk$ information nodes; that is, 
     \[
     \condH{Y_{[\Fr]^t}^F }{ X_{[\Fk]^t}}\;=\;0\;.
     \]

    \item For any subsets $\cA \subseteq [\Fr]^t, \cB \subseteq [\Fk]^t$ such that $\lvert\cA\rvert + \lvert\cB\rvert\;=\;\Fk$, the MDS property ensures:
    \[
    \condH{X_{[\Fk]^t}}{Y_{\cA}^F, X_{\cB}}\;=\;0\;.
    \]

    \item  The MDS property also ensures each final codeword is independent of the others; that is,
    \[
    \Iinfo{X_{[\Ik] \setminus [\Fk]^t}}{X_{[\Fk]^t}}\;=\;0\;.
    \]
\end{enumerate}

Lastly, we assume that all the nodes have a storage capacity of $\alpha$ symbols. 

\subsection{Conversion Coordinator and Bandwidth Analysis}
In the following lemma, we show that all systematic MDS convertible codes are stable in the split regime. The proof follows directly from the MDS property; however, for completeness, we provide a formal proof.

\begin{lemma}\label{lem:stability}
All systematic $(\In, \Ik; \Fn, \Fk)$-MDS convertible codes are stable in the split regime.
\end{lemma}

\begin{proof}

Let $\cN$ and $\cU$ denote the sets of indices of the new nodes and unchanged nodes, respectively. From Remark~\ref{rem:sys_unchanged}, we have:
\[
    \cN \subseteq p^F_{[\Fl \Fr]} \quad \text{and} \quad \cU \supseteq s_{[\Fl \Fk]} \;.
\]
Note that to prove the lemma, it suffices to show that all the initial parity nodes must be retired, i.e.,
\[
    p^I_{[\Ir]} \cap \cU\;=\;\varnothing\;,
\]
which then implies $ \cN\;=\;p^F_{[\Fl \Fr]} $ and $ \cU\;=\;s_{[\Fl \Fk]} $.

Assume, for contradiction, that there exists an index $i \in [\Ir]$ such that $p^I_i \in \cU$. Then, there must exist a tuple $(t, j) \in [\Fl] \times [\Fr]$ such that:
\[
    Y^F_{(t-1)\Fr + j}\;=\;Y^I_i\;.
\]
Now, consider the mutual information between the random variable corresponding to the $(t-1)\Fr + j$-th final parity node and the random variables corresponding to the information nodes of the $t$-th codeword:
\[
    \Iinfo{X_{[\Fk]^t}}{Y^F_{(t-1)\Fr + j}}\;=\;\Iinfo{X_{[\Fk]^t}}{Y^I_i}\;.
\]
From the MDS property of the final code, we have
\begin{align*}
     \Iinfo{X_{[\Fk]^t}}{Y^F_{(t-1)\Fr + j}}\;=\;\Hent{Y^F_{(t-1)\Fr + j}} - \condH{Y^F_{(t-1)\Fr + j}}{X_{[\Fk]^t}}\;=\;\alpha\;.
\end{align*}   
However, for the initial MDS code, as $\Fk + 1 \le 2\Fk \le \Ik$,
\[
    \Iinfo{X_{[\Fk]^t}}{Y^I_i}\;=\;0\;.
\]
This leads to a contradiction. Therefore, we conclude that
\[
    p^I_{[\Ir]} \cap \cU\;=\;\varnothing\;.
\]
This completes the proof.
\end{proof}

Consequently, the conversion coordinator downloads data from the initial storage nodes to construct the new nodes, which by Lemma~\ref{lem:stability}, are exactly the final parity nodes.  Moreover, all initial parity nodes serve as retired nodes. 

To model the conversion process, we introduce the following random variables:
\begin{enumerate}
    \item $V_j$: the data downloaded from the $j$-th initial information node, for $j \in [\Ik]$.
    \item $U_i$: the data downloaded from the $i$-th initial parity node, for $i \in [\Ir]$.
\end{enumerate}
The following properties hold for these random variables:
\begin{enumerate}
    \item Each downloaded random variable is a deterministic function of its stored data:
    \[
      \Hent{V_j \,\bvert\, X_j}\;=\;0\;,
      \quad\;
      \Hent{U_i \,\bvert\, Y_i}\;=\;0\;.
    \]

    \item The entropy of each downloaded random variable is bounded by the per-node storage capacity $\alpha$:
    \[
      \Hent{V_j}\;\triangleq\;\beta_j\;\leq\;\alpha\;,
      \quad\;
      \Hent{U_i}\;\triangleq\;\sigma_i\;\leq\;\alpha\;.
    \]
 \item \textbf{Conversion Coordinator Property:} The final parity nodes are deterministically generated by the conversion coordinator, hence it follows that:
    \[
    \condH{Y_{[\Fl\Fr]}^F}{V_{[\Ik]}, U_{[\Ir]}}\;=\;0\;.
    \]
\end{enumerate}

\noindent Note that the read bandwidth cost, denoted by $\bw$, is the total number of symbols downloaded by the conversion coordinator during the conversion process.

\begin{definition}\label{def:bw_lb}
The read bandwidth cost of any systematic convertible code is defined as
\[
    \bw\;:=\;\sum_{i=1}^{\Ir} \sigma_i +\sum_{j=1}^{\Ik} \beta_j\;.
\]
\end{definition}

The following proposition follows directly from the independence of the data stored across different final codewords.

\begin{proposition}\label{prop:cond_entropy_final}
For any $\cS \subseteq [\Fl]$, conditional entropy of the final parity data given the downloaded data from the information nodes of any systematic convertible code satisfies:
\[
    \condH{Y^F_{[\Fr]^\cS}}{V_{[\Fk]^\cS}}
    \;=\;
    \sum_{t\in \cS} \condH{Y^F_{[\Fr]^t}}{V_{[\Fk]^t}}.
\]
\end{proposition}

\begin{proof}
    Enumerate $\cS = \{i_j \mid j \in [|\cS|]\}$. It follows that, as desired,
\begin{align*}
    \condH{Y^F_{[\Fr]^\cS}}{V_{[\Fk]^\cS}}\;=\;&\sum_{j=1}^{|\cS|}\condH{Y^F_{[\Fr]^{i_j}}}{V_{[\Fk]^{\cS}},Y^F_{[\Fr]^{i_{[j-1]}}}}\tag{Chain rule}\\
    \;=\;&\sum_{j=1}^{|\cS|}\condH{Y^F_{[\Fr]^{i_j}}}{V_{[\Fk]^{i_j}}}\tag{Independence of final codewords}\\
    \;=\;&\sum_{t\in \cS} \condH{Y^F_{[\Fr]^t}}{V_{[\Fk]^t}}\;.
\end{align*}
\end{proof}

First, we start by proving that the bandwidth cost of conversion is lower bounded by the entropy of the data stored in the final parity nodes, or equivalently, the new nodes that are generated as a result of the conversion process.

\begin{lemma}\label{lem:bw_lb_trivial}
The read bandwidth cost of any systematic $(\In, \Ik=\Fl\Fk; \Fn, \Fk)$-MDS convertible code in the split regime satisfies:
    \[
    \bw \ge \Fl\min\{\Fk,\Fr\}\alpha\;.
    \]
\end{lemma}

\begin{proof}
Observe that, as desired,
    \begin{align*}
    \bw\;=\;&\sum_{i=1}^{\Ir} \sigma_i +\sum_{j=1}^{\Ik} \beta_j \tag{Definition~\ref{def:bw_lb}} \\
    \;\ge\;&\Hent{U_{[\Ir]}}+\Hent{V_{[\Ik]}}\\
    \;\ge\;&\Hent{U_{[\Ir]},V_{[\Ik]}}\\
    \;\geq\;&\Iinfo{U_{[\Ir]},V_{[\Ik]}}{Y^F_{[\Fl\Fr]}}\\
    \;=\;&\Hent{Y^F_{[\Fl\Fr]}} - \condH{Y^F_{[\Fl\Fr]}}{U_{[\Ir]}V_{[\Ik]}}\\
    \;=\;&\Hent{Y^F_{[\Fl\Fr]}}\tag{Conversion Coordinator Property}\\
    \;=\;&\sum_{t \in [\Fl]}\Hent{Y^F_{[\Fr]^t}}\tag{Independence of final codewords}\\
    \;=\;&\Fl\min\{\Fk,\Fr\}\alpha\;.\tag{MDS Property}
    \end{align*}
\end{proof}

For the regime $\Fr \ge \Fk$, the bound in Lemma~\ref{lem:bw_lb_trivial} is tight as it matches the bandwidth cost of conversion via the default approach. Therefore, we restrict our focus for the rest of this work to the regime $\Fr < \Fk$. 

Next, we derive a lower bound on the sum of a fraction of $\Hent{V_{[\Ik]}}$ and $\Hent{U_{[\Ir]}}$. This will be used to lower bound the bandwidth cost of conversion with a function that scales with the entropy of the data downloaded from the information nodes.

\begin{lemma}\label{lem:HV_HU_lb}
    When $\Fr < \Fk$, it holds that for any systematic $(\In, \Ik=\Fl\Fk; \Fn, \Fk)$-MDS convertible code in the split regime,
    \[
    \frac{\Fr}{\Fk} \Hent{V_{[\Ik]}} + \Hent{U_{[\Ir]}}  \;\geq\; \Fl\Fr\alpha\;.
    \]
\end{lemma}

\begin{proof}
    Observe that, as desired,
    \begin{align*}
        \frac{\Fr}{\Fk} \Hent{V_{[\Ik]}} + \Hent{U_{[\Ir]}}\;\geq\;&\frac{\Fr}{\Fk} \Hent{V_{[\Ik]}} + \condI{U_{[\Ir]}}{Y^F_{[\Fl\Fr]}}{V_{[\Ik]}}\\
        \;=\;&\frac{\Fr}{\Fk} \Hent{V_{[\Ik]}} + \condI{U_{[\Ir]}}{Y^F_{[\Fl\Fr]}}{V_{[\Ik]}} \tag{Conversion Coordinator Property}\\
        &\qquad+ \condH{Y^F_{[\Fl\Fr]}}{U_{[\Ir]},V_{[\Ik]}}\\
        \;=\;&\frac{\Fr}{\Fk} \Hent{V_{[\Ik]}} + \condH{Y^F_{[\Fl\Fr]}}{V_{[\Ik]}}\\
        \;=\;&\frac{\Fr}{\Fk} \Hent{V_{[\Ik]}} + \sum_{t \in [\Fl]}\condH{Y^F_{[\Fr]^t}}{V_{[\Fk]^t}}\tag{Proposition~\ref{prop:cond_entropy_final}}\\
        \;=\;&\frac{\Fr}{\Fk} \Hent{V_{[\Ik]}} + \sum_{t \in [\Fl]}\left[\Hent{Y^F_{[\Fr]^t}}- \Iinfo{Y^F_{[\Fr]^t}}{V_{[\Fk]^t}}\right]\\
        \;\geq\;&\frac{\Fr}{\Fk} \Hent{V_{[\Ik]}} + \sum_{t \in [\Fl]}\left[\Fr\alpha
        - \frac{\Fr}{\Fk}\Hent{V_{[\Fk]^t}}\right]\tag{MDS Property, Corollary~\ref{cor:mutual_information_avg_bound_1_variable}}\\
        \;=\;&\Fl\Fr\alpha\tag{Independence of $V_{[\Ik]}$}\;.
    \end{align*}
\end{proof}

We proceed with lower bounding the bandwidth cost of conversion by a function of the entropy of the data downloaded from the information nodes, in the regime $\Fr < \Fk$.

\begin{lemma}\label{lem:bw_lb_HV}
    When $\Fr < \Fk$, the read bandwidth cost of any systematic $(\In, \Ik=\Fl\Fk; \Fn, \Fk)$-MDS convertible code in the split regime satisfies:
    \[
    \bw\;\ge\; \frac{\Fk - \Fr}{\Fk}\Hent{V_{[\Ik]}}+ \Fl\Fr\alpha\;.
    \]
\end{lemma}
\begin{proof}
    Observe that, as desired,
    \begin{align*}
        \bw\;=\;&\sum_{i=1}^{\Ir} \sigma_i +\sum_{j=1}^{\Ik} \beta_j\tag{Definition~\ref{def:bw_lb}}\\
        \;\ge\;&\Hent{U_{[\Ir]}} + \Hent{V_{[\Ik]}} \\
        \;=\;&\frac{\Fk - \Fr}{\Fk}\Hent{V_{[\Ik]}} + \frac{\Fr}{\Fk} \Hent{V_{[\Ik]}} + \Hent{U_{[\Ir]}}\\
        \;\ge\;&\frac{\Fk - \Fr}{\Fk}\Hent{V_{[\Ik]}}+ \Fl\Fr\alpha\;.\tag{Lemma~\ref{lem:HV_HU_lb}}
    \end{align*}
\end{proof}

Using Lemma~\ref{lem:HV_HU_lb} and Lemma~\ref{lem:bw_lb_HV},  we derive a lower bound on the bandwidth cost of conversion for the regime $\Fr < \Fk$.

\begin{lemma}\label{lem:bw_lb_I}
When $\Fr < \Fk$, the read bandwidth cost of any systematic $(\In, \Ik=\Fl\Fk; \Fn, \Fk)$-MDS convertible code in the split regime satisfies:
\begin{align*}
    \bw\;\geq\; \Fl\Fk\alpha-\min\{\Ir,\Ik\}\alpha\left(\frac{\Fk}{\Fr}-1\right) \;.
\end{align*} 
\end{lemma}
\begin{proof}
Observe that, as desired,
\begin{align*}
    \bw\;\geq\;&\frac{\Fk - \Fr}{\Fk}\Hent{V_{[\Ik]}}+ \Fl\Fr\alpha\tag{Lemma~\ref{lem:bw_lb_HV}}\\
    \;=\;&\left(\frac{\Fk}{\Fr}-1\right)\left[\frac{\Fr}{\Fk} \Hent{V_{[\Ik]}}\right]+ \Fl\Fr\alpha\\
    \;=\;&\left(\frac{\Fk}{\Fr}-1\right)\left[\frac{\Fr}{\Fk} \Hent{V_{[\Ik]}}+\Hent{U_{[\Ir]}}-\Hent{U_{[\Ir]}} \right]+ \Fl\Fr\alpha\\
    \;\geq\;&\left(\frac{\Fk}{\Fr}-1\right)\left[\Fl\Fr\alpha-\min\{\Ir,\Ik\}\alpha \right]+ \Fl\Fr\alpha\tag{Lemma~\ref{lem:HV_HU_lb}, MDS Property}\\
    \;=\;&\Fl\Fk\alpha-\min\{\Ir,\Ik\}\alpha\left(\frac{\Fk}{\Fr}-1\right).
\end{align*}
\end{proof}

As a consequence of Lemma~\ref{lem:bw_lb_trivial} and Lemma~\ref{lem:bw_lb_I}, the following corollary matches the lower bound of
Maturana and Rashmi~\cite[Theorem 4]{maturana2022bandwidth} which was established under the Uniform Cost Assumption (Definition~\ref{def:uni_cost}).

\begin{corollary}
   The read bandwidth cost of any systematic $(\In, \Ik=\Fl\Fk; \Fn, \Fk)$-MDS convertible code in the split regime satisfies:
\[
\bw \ge 
\begin{cases}
    \Fl\Fk\alpha-\Ir\alpha\max\{\frac{\Fk}{\Fr}-1,0\} & \text{if } \Ir \leq \Fl\Fr\;,\\[8pt]
    \Fl\min\{\Fr,\Fk\}\alpha & \text{otherwise}\;. 
\end{cases}
\]
\end{corollary}

Moreover, for the regime $\Ir \leq \Fr < \Fk$, the bound in Lemma~\ref{lem:bw_lb_I} is tight as it matches the bandwidth cost of conversion of the code constructions from \cite{maturana2022bandwidth}. Therefore, we restrict our focus for the rest of this work to the regime $\Fr < \Fk$ and $\Fr < \Ir$.  

Next, we obtain a lower bound on the entropy of the data downloaded from the information nodes, in the regime $\Fr < \Fk$ and $\Fr < \Ir$. We will combine this result with Lemma~\ref{lem:bw_lb_HV} to obtain a lower bound on the bandwidth cost of conversion as a function of only the parameters of the convertible code.

\begin{lemma}\label{lem:cond_entropy_UVY}
Let $\Fr < \Fk$ and $\Fr < \Ir$. Then, for any $\theta_1 \in [\Fl]$ and $\theta_2 := \max\{0, \Ir  - \theta_1 \Fk\}$, the entropy of the data downloaded from the information nodes of any systematic $(\In, \Ik=\Fl\Fk; \Fn, \Fk)$-MDS convertible code in the split regime satisfies:
\begin{align*}
\Hent{V_{[\Ik]}} \;\geq\; \Fl\Fk\alpha\frac{((\Fl-\theta_1)\Fr-\theta_2)}{((\Fl-\theta_1)\Fr-\theta_2) + \Ir}\;.
\end{align*}
\end{lemma}

\begin{proof}
We begin by observing that
\begin{align*}
    0\;=\;&\condH{Y^F_{[\Fl\Fr]}}{U_{[\Ir]},V_{[\Ik]}}\tag{Conversion Coordinator Property}\\
    \;=\;&\condH{Y^F_{[\Fl\Fr]}}{V_{[\Ik]}}-\condI{Y^F_{[\Fl\Fr]}}{U_{[\Ir]}}{V_{[\Ik]}}\\
    \;\ge\;&\condH{Y^F_{[\Fl\Fr]}}{V_{[\Ik]}}-\condI{Y^F_{[\Fl\Fr]}}{U_{[\Ir]}}{V_{[\Ik]}} - \Iinfo{U_{[\Ir]}}{V_{[\Ik]}}\\
    \;=\;&\condH{Y^F_{[\Fl\Fr]}}{V_{[\Ik]}}-\Iinfo{U_{[\Ir]}}{Y^F_{[\Fl\Fr]},V_{[\Ik]}}\;.
\end{align*}

Since both the final parity nodes and the data downloaded from the information nodes of a final codeword are deterministic functions of the data stored on the information nodes, for any $\cS \subseteq [\Fl]$ of size $\theta_1$,
\begin{align*}
    0\;\ge\;&\condH{Y^F_{[\Fl\Fr]}}{V_{[\Ik]}}-\Iinfo{U_{[\Ir]}}{Y^F_{[\Fr]^{\cS}},V_{[\Fk]^{\cS}},X_{[\Ik]\setminus[\Fk]^{\cS}}}\tag{Data Processing Inequality}\\
    \;=\;&\condH{Y^F_{[\Fl\Fr]}}{V_{[\Ik]}}-\Iinfo{U_{[\Ir]}}{V_{[\Fk]^{\cS}}, X_{[\Ik]\setminus[\Fk]^{\cS}}} - \condI{U_{[\Ir]}}{Y^F_{[\Fr]^{\cS}}}{V_{[\Fk]^{\cS}}, X_{[\Ik]\setminus[\Fk]^{\cS}}}\\
    \;\geq\;&\condH{Y^F_{[\Fl\Fr]}}{V_{[\Ik]}}-\Iinfo{U_{[\Ir]}}{V_{[\Fk]^{\cS}}, X_{[\Ik]\setminus[\Fk]^{\cS}}} - \condH{Y^F_{[\Fr]^{\cS}}}{V_{[\Fk]^{\cS}}, X_{[\Ik]\setminus[\Fk]^{\cS}}}\\
    \;\geq\;&\condH{Y^F_{[\Fl\Fr]}}{V_{[\Ik]}}-\Iinfo{U_{[\Ir]}}{V_{[\Fk]^{\cS}}, X_{[\Ik]\setminus[\Fk]^{\cS}}} - \condH{Y^F_{[\Fr]^{\cS}}}{V_{[\Fk]^{\cS}}}\\
    \;=\;&\sum_{w \in [\Fl]}\condH{Y^F_{[\Fr]^w}}{V_{[\Fk]^w}}-\Iinfo{U_{[\Ir]}}{V_{[\Fk]^{\cS}}, X_{[\Ik]\setminus[\Fk]^{\cS}}} - \sum_{w \in \cS}\condH{Y^F_{[\Fr]^w}}{V_{[\Fk]^w}}\tag{Proposition~\ref{prop:cond_entropy_final}}\\
    \;=\;&\sum_{w \in [\Fl]\setminus\cS}\condH{Y^F_{[\Fr]^w}}{V_{[\Fk]^w}}-\Iinfo{U_{[\Ir]}}{V_{[\Fk]^{\cS}}, X_{[\Ik]\setminus[\Fk]^{\cS}}}\\
    \;\geq\;&\sum_{w \in [\Fl]\setminus\cS}\condH{Y^F_{[\Fr]^w}}{V_{[\Fk]^w}}-\frac{\theta_2}{\Ir}\sum_{i \in [\Ir]}\Hent{U_{i}}-\frac{\Ir - \theta_2}{\theta_1\Fk}\Hent{V_{[\Fk]^{\cS}}}\tag{Corollary~\ref{cor:mutual_information_avg_bound_2_variable}}\\
    \;\geq\;&\sum_{w \in [\Fl]\setminus\cS}\condH{Y^F_{[\Fr]^w}}{V_{[\Fk]^w}}-\theta_2\alpha-\frac{\Ir - \theta_2}{\theta_1\Fk}\Hent{V_{[\Fk]^{\cS}}}\\
    \;=\;&\sum_{w \in [\Fl]\setminus\cS}\condH{Y^F_{[\Fr]^w}}{V_{[\Fk]^w}}-\theta_2\alpha-\frac{\Ir - \theta_2}{\theta_1\Fk}\sum_{w \in \cS}\Hent{V_{[\Fk]^w}}\tag{Independence of $V_{[\Ik]}$}\\
    \;=\;&\sum_{w \in [\Fl]\setminus\cS}\left[\Hent{Y^F_{[\Fr]^w}} - \Iinfo{Y^F_{[\Fr]^w}}{V_{[\Fk]^w}}\right]-\theta_2\alpha-\frac{\Ir - \theta_2}{\theta_1\Fk}\sum_{w \in \cS}\Hent{V_{[\Fk]^w}}\\
    \;\geq\;&\sum_{w \in [\Fl]\setminus\cS}\left[\Fr\alpha - \frac{\Fr}{\Fk}\Hent{V_{[\Fk]^w}}\right]-\theta_2\alpha-\frac{\Ir - \theta_2}{\theta_1\Fk}\sum_{w \in \cS}\Hent{V_{[\Fk]^w}}\tag{MDS Property, Corollary~\ref{cor:mutual_information_avg_bound_1_variable}}\\
    \;=\;&((\Fl - \theta_1)\Fr - \theta_2)\alpha - \frac{\Fr}{\Fk}\sum_{w \in [\Fl]}\mathbb{I}_{w \notin \cS}\Hent{V_{[\Fk]^w}}-\frac{\Ir - \theta_2}{\theta_1\Fk}\sum_{w \in [\Fl]}\mathbb{I}_{w \in \cS}\Hent{V_{[\Fk]^w}}\;.
\end{align*}
If we let $\cA_{\theta_1}$ denote the set of all subsets of $[\Fl]$ of size $\theta_1$, it follows that
\begin{align*}
    0\;\geq\;&\max_{\cS \in \cA_{\theta_1}}\left[((\Fl - \theta_1)\Fr - \theta_2)\alpha - \frac{\Fr}{\Fk}\sum_{w \in [\Fl]}\mathbb{I}_{w \notin \cS}\Hent{V_{[\Fk]^w}}-\frac{\Ir - \theta_2}{\theta_1\Fk}\sum_{w \in [\Fl]}\mathbb{I}_{w \in \cS}\Hent{V_{[\Fk]^w}}\right]\\
    \;\geq\;&\E_{\cS\sim\mathsf{Unif}(\cA_{\theta_1})}\left[((\Fl - \theta_1)\Fr - \theta_2)\alpha - \frac{\Fr}{\Fk}\sum_{w \in [\Fl]}\mathbb{I}_{w \notin \cS}\Hent{V_{[\Fk]^w}}-\frac{\Ir - \theta_2}{\theta_1\Fk}\sum_{w \in [\Fl]}\mathbb{I}_{w \in \cS}\Hent{V_{[\Fk]^w}}\right]\\
    \;=\;&((\Fl - \theta_1)\Fr - \theta_2)\alpha - \frac{\Fr}{\Fk}\sum_{w \in [\Fl]}\frac{\Fl - \theta_1}{\Fl}\Hent{V_{[\Fk]^w}}-\frac{\Ir - \theta_2}{\theta_1\Fk}\sum_{w \in [\Fl]}\frac{\theta_1}{\Fl}\Hent{V_{[\Fk]^w}}\\
    \;=\;&((\Fl - \theta_1)\Fr - \theta_2)\alpha - \frac{((\Fl - \theta_1)\Fr - \theta_2) + \Ir}{\Fl\Fk}\Hent{V_{[\Ik]}}\;.\tag{Independence of $V_{[\Ik]}$}
\end{align*}
    Upon simple algebraic manipulation, we obtain the desired result.
\end{proof}

Observe that when $\Ir > \Ik$ and $\Fk > \Fr$, for any $\theta_1 \in [\Fl]$, $\theta_2 := \max\{0,  \Ir - \theta_1 \Fk\} =  \Ir - \theta_1 \Fk$ and  
\begin{align*}
    (\Fl-\theta_1)\Fr-\theta_2 &= (\Fl-\theta_1)\Fr-\Ir + \theta_1 \Fk\\
    &< (\Fl-\theta_1)\Fr-\Fl\Fk + \theta_1 \Fk\\
    &=  (\Fl-\theta_1)(\Fr-\Fk) \leq 0. 
\end{align*}
Hence, the lower bound in Lemma~\ref{lem:cond_entropy_UVY} is trivial when $\Ir > \Ik$ and $\Fk > \Fr$.

Using Lemma~\ref{lem:bw_lb_HV} and Lemma~\ref{lem:cond_entropy_UVY} and by setting $\theta_1 := \ceil{\frac{\Ir-(\Fr-1)}{\Fk}}$,  we derive a lower bound for the regime $\Fr < \Fk$ and $\Fr < \Ir \leq \Ik$.

\begin{lemma}\label{lem:bw_lb_II}
Let $\Fr < \Fk$ and $\Fr < \Ir \leq \Ik$. Then, the read bandwidth cost of any systematic $(\In,\Ik=\Fl\Fk;\,\Fn,\Fk)$-MDS convertible code in the split regime satisfies
\begin{align*}
\bw\;\geq\;\begin{cases}
            \Fl \Fr \alpha \frac{\left(\Fl - \ceil{\frac{\Ir}{\Fk}}\right)\Fk + \Ir}{\left(\Fl - \ceil{\frac{\Ir}{\Fk}}\right)\Fr + \Ir} &  \text{if}\quad (\Ir \bmod \Fk) \geq \Fr\;, \\
            \Fl\Fk\alpha - \Fl\Ir\alpha\frac{\Fk - \Fr}{\left(\Fl-\floor{\frac{\Ir}{\Fk}}\right)\Fr + \floor{\frac{\Ir}{\Fk}}\Fk} & \text{otherwise}\;.
        \end{cases}
\end{align*}
\end{lemma}

\begin{proof}
Let $b := \Ir \bmod \Fk$. We set $\theta_1 := \ceil{\frac{\Ir-(\Fr-1)}{\Fk}}$, or equivalently,
\[
\theta_1 :=
\begin{cases}
\ceil{\dfrac{\Ir}{\Fk}} & \text{if } b \ge \Fr,\\[4pt]
\floor{\dfrac{\Ir}{\Fk}} & \text{otherwise,}
\end{cases}
\]
and since 
$\theta_2 := \max\{0,\, \Ir - \theta_1 \Fk\}$, it follows that
\[
\theta_2 =
\begin{cases}
0 & \text{if } b \ge \Fr,\\[4pt]
b & \text{otherwise.}
\end{cases}
\]

Hence, the lower bound in Lemma~\ref{lem:cond_entropy_UVY} reduces to:
    \begin{align}
\Hent{V_{[\Ik]}} \;\geq\; \begin{cases}
    \Fl\Fk\alpha\frac{\left(\Fl - \ceil{\frac{\Ir}{\Fk}}\right)\Fr}{\left(\Fl - \ceil{\frac{\Ir}{\Fk}}\right)\Fr + \Ir}  &  \text{if}~ b \geq \Fr\;, \\
    \Fl\Fk\alpha\frac{\left(\Fl - \floor{\frac{\Ir}{\Fk}}\right)\Fr - b}{\left(\Fl - \floor{\frac{\Ir}{\Fk}}\right)\Fr + \Ir - b} & \text{otherwise}\;.
\end{cases}\label{ineq:HV}
\end{align}
It follows that, as desired,
\begin{align*}
        \bw&\;\ge\; \frac{\Fk - \Fr}{\Fk}\Hent{V_{[\Ik]}}+ \Fl\Fr\alpha\tag{Lemma~\ref{lem:bw_lb_HV}}\\
        &\;\ge\;\begin{cases}
            \Fl(\Fk - \Fr)\alpha\frac{\left(\Fl - \ceil{\frac{\Ir}{\Fk}}\right)\Fr}{\left(\Fl - \ceil{\frac{\Ir}{\Fk}}\right)\Fr + \Ir} + \Fl\Fr\alpha
             &\text{if}~ b \geq \Fr\;, \\
             \Fl(\Fk-\Fr)\alpha\frac{\left(\Fl - \floor{\frac{\Ir}{\Fk}}\right)\Fr - b}{\left(\Fl - \floor{\frac{\Ir}{\Fk}}\right)\Fr + \Ir - b} + \Fl\Fr\alpha& \text{otherwise}\;.
        \end{cases}\tag{Inequality~\eqref{ineq:HV}}\\
        &\;=\;\begin{cases}
            \Fl\Fr\alpha\frac{\left(\Fl - \ceil{\frac{\Ir}{\Fk}}\right)(\Fk - \Fr)}{\left(\Fl - \ceil{\frac{\Ir}{\Fk}}\right)\Fr + \Ir} + \Fl\Fr\alpha
             &\text{if}~ b \geq \Fr\;, \\
             \Fl(\Fk-\Fr)\alpha\left(1-\frac{\Ir}{\left(\Fl - \floor{\frac{\Ir}{\Fk}}\right)\Fr + \floor{\frac{\Ir}{\Fk}}\Fk}\right) + \Fl\Fr\alpha& \text{otherwise}\;.
        \end{cases}\\
        &\;=\;\begin{cases}
            \Fl \Fr \alpha \frac{\left(\Fl - \ceil{\frac{\Ir}{\Fk}}\right)\Fk + \Ir}{\left(\Fl - \ceil{\frac{\Ir}{\Fk}}\right)\Fr + \Ir} & \Ir \leq (\Fl-1)\Fk\;, \\
            \Fl\Fk\alpha - \Fl\Ir\alpha\frac{\Fk - \Fr}{\left(\Fl-\floor{\frac{\Ir}{\Fk}}\right)\Fr + \floor{\frac{\Ir}{\Fk}}\Fk} & \text{otherwise}\;.
        \end{cases}
\end{align*}

\end{proof}

For the regime $\Fr < \Fk$ and $\Fr < \Ir \leq \Fk$, the bound in Lemma~\ref{lem:bw_lb_II} is tight as it matches the bandwidth cost of conversion of the code constructions from \cite{maturana2022bandwidth}. Consequently, Lemma~\ref{lem:bw_lb_II} resolves Conjecture~\ref{conj:split_regime} for systematic convertible codes for the regime $\Fr < \Fk$ and $\Fr < \Ir \leq \Fk$.

Finally, by combining Lemma~\ref{lem:bw_lb_trivial}, Lemma~\ref{lem:bw_lb_I}, and Lemma~\ref{lem:bw_lb_II}, we obtain the following theorem whose proof can be found in Appendix. 

\begin{restatable}{theorem}{MDSSplitLb}\label{thm:MDS_split_lb}
The read bandwidth cost of any systematic $(\In,\Ik=\Fl\Fk;\,\Fn,\Fk)$-MDS convertible code in the split regime satisfies
\begin{align*}
    \bw\;\geq\;\begin{cases}
        \Fl \Fr \alpha \frac{\left(\Fl - \ceil{\frac{\Ir}{\Fk}}\right)\Fk + \Ir}{\left(\Fl - \ceil{\frac{\Ir}{\Fk}}\right)\Fr + \Ir} & \text{~if}~ (\Ir \bmod \Fk) \geq \Fr,\;\Fr < \Ir \leq \Ik,\;\Fr < \Fk\;, \\
        \Fl\Fk\alpha - \Fl\Ir\alpha\frac{\Fk - \Fr}{\left(\Fl-\floor{\frac{\Ir}{\Fk}}\right)\Fr + \floor{\frac{\Ir}{\Fk}}\Fk} & \text{~if}~ (\Ir \bmod \Fk) < \Fr,\;\Fr < \Ir \leq \Ik,\;\Fr < \Fk\;, \\
\Fl\Fk\alpha-\Ir\alpha\left(\frac{\Fk}{\Fr}-1\right) & \text{ if}~\Ir \le \Fr < \Fk\;,\\
\Fl \min\{\Fk,\Fr\} \alpha & \text{~otherwise}\;.
    \end{cases} 
\end{align*}
 Moreover, this bound is tight when $\Fr \geq \Fk$ or $\Ir \leq \Fk$.
\end{restatable}

\section{Conclusion}\label{sec:conclusion}
In this work, we established fundamental lower bounds on the bandwidth cost of systematic MDS convertible codes in the split regime. The previously best known bounds were derived under the assumption that a uniform amount of data is downloaded across the data nodes (and similarly for the parity nodes). Further, these bounds were tight for the parameter regime where $\Fr \geq \Fk$ or $\Ir \leq \Fr$. To the contrary, in this work, we derived lower bounds without any such assumption. Moreover, our lower bounds are tight for the broader parameter regime where $\Fr \geq \Fk$ or $\Ir \leq \Fk$. In addition, as a consequence of this work, we partially resolve the conjecture of Maturana and Rashmi in~\cite{maturana2022bandwidth}.

\paragraph*{Open Problems and Future Work}
A natural avenue for future work is to extend the techniques developed in this paper to close the gap between the lower and upper bounds for the parameter regime where $\Ir > \Fk > \Fr$, which would likely necessitate a re-examination of the proof of Lemma~\ref{lem:cond_entropy_UVY}.

A broader avenue of research is to study the bandwidth cost of MDS convertible codes—including both lower bounds and code constructions—for the general regime, which remains largely unexplored. Another direction of future research is to study the bandwidth cost of locally repairable convertible codes using the information-theoretic techniques developed in this work.

\newpage
\bibliographystyle{plain}
\bibliography{ref}

\newpage
\appendix
\MDSSplitLb*
\begin{proof} 
Since the lower bounds from Lemma~\ref{lem:bw_lb_trivial} and  Lemma~\ref{lem:bw_lb_I} are tight when 
$\Fr \geq \Fk$ and $\Ir \leq \Fr < \Fk$, respectively, we restrict our focus to the regime $\Fr < \Ir$ and $\Fr < \Fk$. Let 
\begin{align*}
    \cL_1 &:= \begin{cases}
        \Fl \Fr \alpha \frac{\left(\Fl - \ceil{\frac{\Ir}{\Fk}}\right)\Fk + \Ir}{\left(\Fl - \ceil{\frac{\Ir}{\Fk}}\right)\Fr + \Ir} & \text{if}~ (\Ir \bmod \Fk) \geq \Fr,\;\Fr < \Ir \leq \Ik,\;\Fr < \Fk\;,\\
        \Fl\Fk\alpha - \Fl\Ir\alpha\frac{\Fk - \Fr}{\left(\Fl-\floor{\frac{\Ir}{\Fk}}\right)\Fr + \floor{\frac{\Ir}{\Fk}}\Fk} & \text{if}~ (\Ir \bmod \Fk) < \Fr,\;\Fr < \Ir \leq \Ik,\;\Fr < \Fk\;,
        \end{cases}\\
    \cL_2 &:= \Fl\Fk\alpha-\min\{\Ir,\Ik\}\alpha\left(\frac{\Fk}{\Fr}-1\right) \;,\\
    \cL_3 &:= \Fl\min\{\Fk,\Fr\}\alpha\;.
\end{align*}

Then, from Lemma~\ref{lem:bw_lb_trivial}, Lemma~\ref{lem:bw_lb_I}, and Lemma~\ref{lem:bw_lb_II}, it follows that
\begin{align*}
    \gamma_r \geq \begin{cases}
        \max\{\cL_1, \cL_2, \cL_3\} & \text{if}~ \Fr < \Ir \leq \Ik, \Fr < \Fk\;,\\
        \max\{\cL_2,\cL_3\} & \text{if}~ \Ik < \Ir, \Fr < \Fk\;.
    \end{cases}
\end{align*}

\noindent
\textbf{Case 1:} $\Fr < \Ir \leq \Ik$ and $\Fr < \Fk$.

\noindent Then, $\min\{\Ir,\Ik\}=\Ir$, so
\begin{align*}
\cL_2 = \Fl\Fk\alpha - \Ir\alpha\left(\frac{\Fk}{\Fr} - 1\right)\;.
\end{align*}
Also, $\cL_3 = \Fl\Fr\alpha$ since $\Fr < \Fk$.

\noindent
\underline{Subcase 1a:} $(\Ir \bmod \Fk) \geq \Fr$. Then, 
\begin{align*}
    \Ir &\geq \floor{\frac{\Ir}{\Fk}}\Fk + \Fr \\
    &= \left(\ceil{\frac{\Ir}{\Fk}}-1\right)\Fk + \Fr \tag{Since $(\Ir \bmod \Fk) > 0$} \\
    &> \ceil{\frac{\Ir}{\Fk}} \Fr\;. \tag{Since $\Fr < \Fk$}
\end{align*}

Since $(\Ir \bmod \Fk) \geq \Fr$, 
\[
\cL_1 = \Fl\Fr\alpha 
\frac{(\Fl - \ceil{\frac{\Ir}{\Fk}})\Fk + \Ir}
     {(\Fl - \ceil{\frac{\Ir}{\Fk}})\Fr + \Ir}\;.
\]
Since $\Ir \leq \Ik$ and $\Fr < \Fk$, 
\[
\cL_1 = \Fl\Fr\alpha 
\frac{(\Fl - \ceil{\frac{\Ir}{\Fk}})\Fk + \Ir}
     {(\Fl - \ceil{\frac{\Ir}{\Fk}})\Fr + \Ir} \geq \Fl\Fr\alpha = \cL_3.
\]
Moreover,
\begin{align*}
    \cL_1 &= \Fl\Fr\alpha \frac{(\Fl - \ceil{\frac{\Ir}{\Fk}})\Fk + \Ir}
     {(\Fl - \ceil{\frac{\Ir}{\Fk}})\Fr + \Ir}\\
&= \Fl\Fr\alpha + \Fl\Fr\alpha\frac{(\Fk-\Fr)(\Fl-\ceil{\frac{\Ir}{\Fk}})}
             {(\Fl-\ceil{\frac{\Ir}{\Fk}})\Fr+\Ir}\\
&>\Fl\Fr\alpha + \Fl\Fr\alpha\frac{(\Fk-\Fr)(\Fl- \frac{\Ir}{\Fr})}
             {(\Fl-\ceil{\frac{\Ir}{\Fk}})\Fr+\ceil{\frac{\Ir}{\Fk}}\Fr}\tag{Since $\Ir > \ceil{\frac{\Ir}{\Fk}} \Fr$ }\\
&= \Fl\Fr\alpha + \alpha(\Fk-\Fr)\left(\Fl- \frac{\Ir}{\Fr}\right)\\
&= \Fl\Fk\alpha - \Ir\alpha\left(\frac{\Fk}{\Fr} - 1\right) = \cL_2\;.
\end{align*}
Hence, $\max\{\cL_1, \cL_2, \cL_3\}=\cL_1$.

\noindent
\underline{Subcase 1b:} $(\Ir \bmod \Fk) < \Fr$. 
 Then, 
 \begin{align*}
     \Ir < \Fr + \floor{\frac{\Ir}{\Fk}}\Fk\;.
 \end{align*}
Further, since $\Ir \leq \Ik$, it follows that
  \begin{align}\label{eq:thm_rI_ub}
     \Ir \leq \left(\Fl - \floor{\frac{\Ir}{\Fk}}\right)\Fr + \floor{\frac{\Ir}{\Fk}}\Fk\;.  
 \end{align}

Since $(\Ir \bmod \Fk) \geq \Fr$, 
\begin{align*}
\cL_1&=\Fl\Fk\alpha - \Fl\Ir\alpha\frac{\Fk - \Fr}{\left(\Fl-\floor{\frac{\Ir}{\Fk}}\right)\Fr + \floor{\frac{\Ir}{\Fk}}\Fk} \\
&>\Fl\Fk\alpha - \Fl\Ir\alpha\frac{\Fk - \Fr}{\left(\Fl-\floor{\frac{\Ir}{\Fk}}\right)\Fr + \floor{\frac{\Ir}{\Fk}}\Fr} \tag{Since $\Fr < \Fk$}\\
& = \Fl\Fk\alpha - \Ir\alpha\left(\frac{\Fk}{\Fr} - 1\right) = \cL_2\;.
\end{align*}
Hence $\cL_1 > \cL_2$. Moreover,
\begin{align*}
\cL_1&=\Fl\Fk\alpha - \Fl\Ir\alpha\frac{\Fk - \Fr}{\left(\Fl-\floor{\frac{\Ir}{\Fk}}\right)\Fr + \floor{\frac{\Ir}{\Fk}}\Fk} \\
&= \Fl(\Fk-\Fr)\alpha\left(1-\frac{\Ir}{\left(\Fl - \floor{\frac{\Ir}{\Fk}}\right)\Fr + \floor{\frac{\Ir}{\Fk}}\Fk}\right) + \Fl\Fr\alpha\\
&\geq \Fl\Fr\alpha \tag{Inequality~\eqref{eq:thm_rI_ub}} = \cL_3\;.
\end{align*}

Hence, $\max\{\cL_1, \cL_2, \cL_3\}=\cL_1$.

\noindent
\textbf{Case 2:} $\Ik < \Ir$ and $\Fr < \Fk$.

\noindent
Then, $\min\{\Ir,\Ik\}=\Ik$, so
\begin{align*}   
\cL_2 &= \Fl\Fk\alpha - \Ik\alpha\frac{\Fk-\Fr}{\Fr}\\
&= \Fl\Fr\alpha + \left(\Fl-\frac{\Ik}{\Fr}\right)(\Fk-\Fr)\alpha\\
&< \Fl\Fr\alpha + \left(\Fl-\frac{\Fl\Fk}{\Fk}\right)(\Fk-\Fr)\alpha \tag{Since $\Ik = \Fl\Fk$ and $\Fr < \Fk$} \\
&= \Fl\Fr\alpha = \cL_3\;.
\end{align*}

Hence, $\max\{\cL_2, \cL_3, \cL_3\}=\cL_3$.

\noindent
Combining both cases completes the proof.
\end{proof}

\end{document}